\documentclass[11pt,a4paper]{article}
\usepackage{graphicx,nicefrac}
\usepackage{tabularx,booktabs,multirow}
\usepackage{todonotes}
\usepackage{amssymb,amsmath,amsthm}
\usepackage{mathtools}
\usepackage[pagebackref]{hyperref}
\usepackage{subcaption}
\usepackage{amsmath,amsfonts}
\usepackage{xcolor}
\usepackage{xspace}
\usepackage{bbm}
\usepackage{tikz}
\usepackage{makecell}
\usepackage{authblk}
\usepackage{fullpage}
\usepackage{graphicx}

\usepackage{amsfonts}

\usepackage[]{hyperref}
\usepackage{mathtools} 
\usepackage{tabularx} 
\usepackage{marvosym} 
\usepackage{bm} 
\usepackage{multirow} 
\usepackage{thmtools, thm-restate} 
\usepackage{amsmath,amssymb}
\usepackage{algorithm}
\usepackage[noend]{algpseudocode}
\algnewcommand\algorithmicinput{\textbf{Input:}}
\algnewcommand\algorithmicoutput{\textbf{Output:}}
\algnewcommand\Input{\item[\algorithmicinput]}
\algnewcommand\Output{\item[\algorithmicoutput]}

\usepackage{etoolbox}
\usepackage{tikz}
\usetikzlibrary{tikzmark}
\usetikzlibrary{arrows, calc, shapes.geometric, decorations.pathreplacing,calligraphy}
\usepackage{subfiles}

\errorcontextlines\maxdimen

\newcommand{\ALGtikzmarkcolor}{black}
\newcommand{\ALGtikzmarkextraindent}{4pt}
\newcommand{\ALGtikzmarkverticaloffsetstart}{-.5ex}
\newcommand{\ALGtikzmarkverticaloffsetend}{-.5ex}
\makeatletter
\newcounter{ALG@tikzmark@tempcnta}

\newcommand\ALG@tikzmark@start{%
	\global\let\ALG@tikzmark@last\ALG@tikzmark@starttext%
	\expandafter\edef\csname ALG@tikzmark@\theALG@nested\endcsname{\theALG@tikzmark@tempcnta}%
	\tikzmark{ALG@tikzmark@start@\csname ALG@tikzmark@\theALG@nested\endcsname}%
	\addtocounter{ALG@tikzmark@tempcnta}{1}%
}

\def\ALG@tikzmark@starttext{start}
\newcommand\ALG@tikzmark@end{%
	\ifx\ALG@tikzmark@last\ALG@tikzmark@starttext
	\else
	\tikzmark{ALG@tikzmark@end@\csname ALG@tikzmark@\theALG@nested\endcsname}%
	\tikz[overlay,remember picture] \draw[\ALGtikzmarkcolor] let \p{S}=($(pic cs:ALG@tikzmark@start@\csname ALG@tikzmark@\theALG@nested\endcsname)+(\ALGtikzmarkextraindent,\ALGtikzmarkverticaloffsetstart)$), \p{E}=($(pic cs:ALG@tikzmark@end@\csname ALG@tikzmark@\theALG@nested\endcsname)+(\ALGtikzmarkextraindent,\ALGtikzmarkverticaloffsetend)$) in (\x{S},\y{S})--(\x{S},\y{E});%
	\fi
	\gdef\ALG@tikzmark@last{end}%
}

\newtheorem{observation}{Observation}

\newtheorem{theorem}{Theorem}

\newtheorem{proposition}{Proposition}

\apptocmd{\ALG@beginblock}{\ALG@tikzmark@start}{}{\errmessage{failed to patch}}
\pretocmd{\ALG@endblock}{\ALG@tikzmark@end}{}{\errmessage{failed to patch}}
\makeatother

\newcommand{\appsymb}{$\star$}
\newcommand{\appref}[1]{{\hyperref[proof:#1]{\appsymb}}}
\newcommand{\appLink}[1]{{\hyperref[#1]{\appsymb}}}

\newcommand{\problemn}[1]{\textsc{#1}}
\newcommand{\problemdef}[3]{
	\begin{center}
		\begin{minipage}{0.95\textwidth}
			\noindent
			\problemn{#1}
			\vspace{5pt}\\
			\setlength{\tabcolsep}{3pt}
			\begin{tabularx}{\textwidth}{@{}lX@{}}
				\textbf{Input:}     & #2 \\
				\textbf{Question:}  & #3
			\end{tabularx}
		\end{minipage}
	\end{center}
}

\newcommand{\NN}{\mathbb{N}}

\newcommand*{\setcard}[1]{\left\lvert #1 \right\rvert}

\newcommand{\bigOh}{\mathcal{O}}
\newcommand*{\symmdiff}{\triangle}

\newcommand*{\enc}[1]{\left\langle #1\right\rangle}

\newcommand*{\sinks}{\operatorname{S}}

\newcommand*{\graphsize}[1]{\lVert #1\rVert}
\newcommand*{\gsubset}{\subseteq^S}

\newcommand*{\delegation}{delegation}

\newcommand*{\election}{election}

\newcommand*{\alternative}{alternative}
\newcommand*{\alternatives}{\alternative{}s}

\newcommand*{\Allv}{\textsc{All}}
\newcommand*{\allv}{\textsc{All}}
\newcommand*{\Onev}{\textsc{One}}
\newcommand*{\onev}{\textsc{One}}

\newcommand*{\windetproblemtext}{Winner Determination}
\newcommand*{\windetproblem}{\problemn{\windetproblemtext}}
\newcommand*{\windetgenerictext}{Majority/Plurality-One/All \windetproblemtext{}}

\newcommand*{\windetallmajtext}{Majority-All \windetproblemtext{}}
\newcommand*{\windetallmaj}{\problemn{\windetallmajtext}}
\newcommand*{\windetallpluralitytext}{Plurality-All \windetproblemtext{}}
\newcommand*{\windetallplurality}{\problemn{\windetallpluralitytext}}
\newcommand*{\windetonemajtext}{Majority-One \windetproblemtext{}}
\newcommand*{\windetonemaj}{\problemn{\windetonemajtext}}
\newcommand*{\windetonepluralitytext}{Plurality-One \windetproblemtext{}}
\newcommand*{\windetoneplurality}{\problemn{\windetonepluralitytext}}

\newcommand*{\electionbriberyproblemtext}{Election Bribery}
\newcommand*{\electionbriberyproblem}{\problemn{\electionbriberyproblemtext}}

\newcommand*{\delegationbriberyproblemtext}{Delegation Bribery}
\newcommand*{\delegationbriberyproblem}{\problemn{\delegationbriberyproblemtext}}
\newcommand*{\delegationbriberygenerictext}{Majority/Plurality \delegationbriberyproblemtext}
\newcommand*{\delegationbriberygeneric}{\problemn{\delegationbriberygenerictext}}
\newcommand*{\delegationbriberymajorityproblemtext}{Majority \delegationbriberyproblemtext{}}
\newcommand*{\delegationbriberymajorityproblem}{\problemn{\delegationbriberymajorityproblemtext}}
\newcommand*{\delegationbriberypluralityproblemtext}{Plurality \delegationbriberyproblemtext{}}
\newcommand*{\delegationbriberypluralityproblem}{\problemn{\delegationbriberypluralityproblemtext}}

\newcommand*{\revbfs}{\operatorname{rev-bfs}}
\newcommand*{\resrevbfs}{\operatorname{res-rev-bfs}}
\newcommand*{\vweight}{vp}

\usepackage[nameinlink,capitalize]{cleveref}
\usetikzlibrary{shapes.geometric}
\crefname{observation}{Observation}{Observations}

\title{Who won?\\Winner Determination and Robustness in Liquid Democracy}

\author{Matthias Bentert} 
\author{Niclas Boehmer}
\author{Maciej Rymar}
\author{Henri Tannenberg}

\affil{\small
  Technische Universit\" at Berlin, Algorithmics and Computational Complexity\protect\\
  \{matthias.bentert,niclas.boehmer,m.rymar,h.tannenberg\}@tu-berlin.de}

\date{\today}

\begin{document}

\maketitle

\begin{abstract}
Liquid democracy is a decision-making paradigm in which each agent can either vote directly for some alternative  or (transitively) delegate its vote to another agent.
To mitigate the issue of delegation cycles or the concentration of power, delegating agents might be allowed to specify multiple delegation options. 
Then, a (cycle-free) delegation is selected in which each delegating agent has exactly one representative.
We study the winner determination problem for this setting, i.e., whether we can select a delegation such that a given alternative wins (or does not win). 
Moreover, we study the robustness of winning alternatives in two ways: 
First, we consider whether we can make a limited number of changes to the preferences cast by the agents such that a given alternative becomes a winner in one/in all delegations, and second, whether we can make a limited number of changes to a selected delegation to make a given alternative a winner.
	
\end{abstract} 

\section{Introduction}\label{sec:intro}
Liquid democracy is a paradigm that aims to change the  average citizen's connection to decision making.
In it, each agent is allowed to either vote directly for some alternative \emph{or} to delegate its vote to any other agent.
Such delegations are transitive---a delegated vote can be delegated again and again, until an agent voting directly for some alternative is reached.
Another important feature of liquid democracy is that participants can adjust their choices on a regular basis~\mbox{\cite{blum2016liquid,ford2002delegative,paulin2020overview}}.

One problem that commonly arises in liquid democracy is that of delegation cycles~\cite{brill2018interactive,golz2018fluid}. 
If a group of agents delegates in a cyclic fashion with none of them voting directly, then their votes cannot be counted and are lost.
One possibility to mitigate this issue, which has been previously suggested for instance by G{\"{o}}lz et al.~\cite{golz2018fluid}, is to allow each delegating agent to specify a set of agents they approve to receive their vote.\footnote{Recently, Brill et al.~\cite{brill2021liquid} studied a similar model where agents have additionally preferences over all delegating agents they approve.}
This can be modeled as a directed \emph{election graph}~$G$ where we have one vertex for each agent and an arc from an agent $a$ to an agent $b$ if $a$ approves $b$.
To model the opinions of directly voting agents, we subsequently extend the election graph by adding a vertex for each alternative and adding an arc from some agent $a$ to an alternative $c$ if $a$ directly votes for~$c$.
Subsequently, a cycle-free \emph{delegation} is selected in which each agent either votes directly, or delegates its vote to exactly one of its approved agents (a delegation is an acyclic subgraph of $G$ where the set of sinks is exactly the set of alternatives).
The selected delegation then ``induces'' an election where each directly voting agent $a$ does not only vote for itself but also for all agents that (transitively) delegated their vote to $a$.
Imagine now that some alternative wins in the election induced by the selected delegation but that there also exists a different delegation where a different alternative wins the induced election.\footnote{This happens, for instance, if there are two alternatives~$c$ and~$d$ and three agents $a_1$, $a_2$, and~$a_3$, where $a_1$ votes for~$c$, $a_2$ votes for $d$, and~$a_3$ approves both delegating to $a_1$ and $a_2$.}
Then, the proponents of the latter alternative may have reasonable grounds to complain: They can argue that their loss did not reflect ``the will of the people,'' but is due to the delegation selection algorithm, which happened to select a delegation that is unfavorable to them.
For the integrity and transparency of the process, it would therefore be desirable to recognize such cases.
Motivated by this, we study the \textsc{All [One] Winner Determination} problem, where we check whether a given alternative wins in \emph{all} [in at least \emph{one}] delegation in a given election graph.

But what should we do if there exist two delegations in an election graph producing different winning alternatives?
One answer to this is to assess which alternative's win is more robust or which of the two alternatives is closer to winning in all possible delegations (we measure closeness here in terms of how many arcs in the election graph need to be changed).
This motivates the study of the \textsc{All [One] Election Bribery} problem, where given an election graph, some alternative $c$ and an integer $k$, the question is whether we can insert or delete $k$ arcs in the election graph such that $c$ is a winner in \emph{all} [in at least \emph{one}] delegation.
Note that both problem variants are framed constructively, i.e., we want that some given alternative wins in certain cases. However, all our results will also hold for the destructive variants, where we want to prevent a certain alternative from winning in at least one/in all delegations.
Thus, to choose between multiple alternatives each winning in some delegation, we can compare the minimum cost of an All Election Bribery for the alternatives and pick the one with the smallest cost. Alternatively, we can compare their minimum cost of a destructive One Election Bribery and pick the one with the highest cost.
The name bribery here stems from the fact that we can alternatively interpret these problems as questions an ``evil'' briber faces when they decide which agents to bribe in order to change the situation in their favor. Notably, while from the robustness perspective fast algorithms are particularly desirable, from the bribery perspective computational hardness results are appealing, as they offer a worst-case computational barrier against optimal manipulation.

Importantly, \textsc{Election Bribery} is also relevant in the case where initially only one alternative wins in all delegations, as  it allows us to quantify the robustness of the win of this alternative by computing the minimum number of changes to change the winner of the election (we can compute this by applying the \textsc{Election Bribery} problem to all non-winning alternatives and then taking the minimum cost).  
This is also related to power accumulation in liquid democracy \cite{kling2015voting,golz2018fluid}.
An example of a very robust win is when all agents vote directly for some alternative~$c$ and an example of a very non-robust win is when all agents delegate their vote to a single agent which then directly votes for~$c$.
Detecting non-robust wins is extremely helpful, as in these cases one needs to be particularly careful regarding an external manipulation of the election and technical errors. 
In such cases, one might even want to further examine the election~issue and advise voters to reconsider their cast preferences.
On a related note, \textsc{Election Bribery} also offers a way to rank alternatives that do not win in any delegation by computing their distance from winning.

A different approach to tackle situations where there are  delegations in an election graph producing different winning alternatives is to examine the robustness of these delegations. 
We formalize the problem of computing the robustness of a delegation as the \textsc{Delegation Bribery} problem, where we are given a delegation and a budget $k$ and the question is whether there is a group of $k$ agents that can change the outcome together by voting directly for some given alternative.\footnote{Again this problem is framed constructively, yet our results also apply to the destructive variant.}
To increase the possibilities for direct participation and as one delegation can be possibly used for multiple election issues, after computing the delegation, it is quite natural to inform agents to which agent their vote is delegated and to allow them to vote directly on the issue instead \cite{LQblog}. 
That delegating agents are unhappy with the agent receiving its vote in one poll can easily happen, as it is common for people to broadly agree with another person on a wide range of political issues, but have different opinions on some highly controversial subjects. 
Moreover, the \textsc{Delegation Bribery} problem is also relevant in the context of other liquid democracy models where, for instance, each delegating agent is asked to specify only one delegation option. 

\paragraph{Related Work.}
While the basic ideas of liquid democracy can be traced back all the way to the 1800s~\cite{carroll1884principles}, the start of the modern study of liquid democracy is often attributed to Ford \cite{ford2002delegative}.
For a more comprehensive review of the history of liquid democracy, we refer to the work of Behrens~\cite{behrens2017origins}.
Currently, liquid democracy is a very active area of research~\cite{brill2021liquid,d2021computation,golz2018fluid,kotsialou2018incentivising,paulin2020overview,zhang2021power} (see also  \cite{behrens2017origins,brill2018interactive,ford2020liquid,paulin2020overview} for surveys). 
While many different liquid democracy models have been proposed, our model closely follows that of
G\"{o}lz et al.\ \cite{golz2018fluid}.
The main difference between our model and the model of G\"{o}lz et al.\ \cite{golz2018fluid} is that we add information about the alternatives the directly voting agents vote for.
Despite this difference, we strengthen one of G\"{o}lz et al.'s hardness results in \cref{subsec:windetexists}.

Concerning the computational problems studied in this paper, to the best of our knowledge, the \windetproblem{} problem has neither been studied in our nor any other model of liquid democracy. 
However, the broader problem of finding delegations with ``good'' properties has attracted a significant amount of research \cite{brill2021liquid,brill2018pairwise,golz2018fluid,kotsialou2018incentivising}.
For our bribery problems, there is one previous work by D'Angelo et al. \cite{d2021computation} analyzing bribery related problems in the context of liquid democracy using a model similar to ours.
In particular, one of their problems is similar to our \delegationbriberyproblem{} problem.
While we allow delegating agents to change their mind and directly vote for an alternative, they only allow delegating agents to switch from one initially approved agent to another initially approved agent when redirecting the delegation of their vote.
D'Angelo et al. \cite[Theorem 5.6]{d2021computation} showed that their version of the problem is NP-hard.
In contrast, we show our version to be polynomial-time solvable.
More generally speaking, bribery and its connection to robustness has been extensively studied in the broader context of computational social choice \cite{DBLP:conf/aldt/BaumeisterER11,DBLP:journals/jair/BoehmerBHN21,DBLP:conf/ijcai/BoehmerBKL20,faliszewski2016control}. 

\paragraph{Our Contribution.}
We study the complexity of \textsc{Winner Determination}, \textsc{Election Bribery}, and \textsc{Delegation Bribery} as informally introduced above and formally introduced in the respective sections.
We consider all problems for two different voting rules used for evaluating the election induced by some delegation: Under the plurality voting rule, an alternative is a winner if it receives the highest number of votes, while under the majority voting rule an alternative is a winner if it receives at least half of the votes.
While Plurality is arguably the most popular voting rule in the political realm, our reason for considering the majority voting rule is two-fold: 
First, it is in some sense simpler than the Plurality voting rule (allowing for additional tractability results). 
Second, for certain (political) decisions an absolute majority is needed, so in these cases the majority voting rule is used.
\cref{tab:allresults} gives an overview over our results.
\begin{table}[t]
	\caption{
		Overview of our results. ``P'' stands for polynomial-time solvable and ``NP-h.'' stands for NP-hardness. All our algorithms run in quadratic or even subquadratic time. All hardness results hold even if the election graph is a directed acyclic graph with maximum outdegree two.
	}
	\label{tab:allresults}
	\centering
	\begin{tabular}{c|cc}
		& \textsc{Majority} & \textsc{Plurality} \\
		\hline
		\allv{} \windetproblem{} & P (Pr. \ref{thm:windetallpluralityeasy}) & P(Pr. \ref{thm:windetallpluralityeasy}) \\
		\onev{} \windetproblem{} & P (Pr. \ref{thm:majoritywinnereasy})& NP-h. for three alternatives (Th. \ref{thm:windetexpluralityhard}) \\
		\hline
		\allv{} \electionbriberyproblem{} & \multicolumn{2}{c}{NP-h. and W[1]-h. wrt. $k$ for two alternatives (Th. \ref{thm:edgeallhard})} \\
		\onev{} \electionbriberyproblem{} & \multicolumn{2}{c}{NP-h. and W[1]-h. wrt. $k$ for two alternatives (Th. \ref{thm:edgeonehard})} \\
		\hline
		\delegationbriberyproblem{} & P (Pr. \ref{thm:delegbribemajeasy}) & P (Th. \ref{th:delegation-plur})
	\end{tabular}
\end{table}%
Before we describe our results in more detail, we want to remark that all our hardness results hold in very restrictive settings. 
In particular, we show hardness even if the election graph is acyclic and each vertex has maximum outdegree two. 
Both of these properties are highly relevant in practice, as an acyclic election graph can, for instance, arise if agents have competence levels and agents only delegate to more qualified agents \cite{kahng2021liquid}. 
The maximum outdegree, on the other hand, describes the maximum number of agents some agent approves as possible delegators, a number which is presumably a very small constant in practice. 

Coming to our results, we study the \textsc{Winner Determination} problem in \Cref{sec:windet}. 
We show that \textsc{Plurality-One Winner Determination} is already NP-hard for three alternatives, whereas the three other problem variants are polynomial-time solvable.
This implies that under the plurality rule, we presumably cannot efficiently check which alternatives can be a winner; however, we can check whether a winning alternative is the only possible winner, a practically already quite relevant question. 
In  \Cref{sec:elecbribe}, we study the \textsc{Election Bribery} problem. 
We show that all four problem variants are already NP-hard for only two alternatives. 
Moreover, we strengthen these hardness results by proving  that the problems are even W[1]-hard with respect to the given budget $k$. 
This implies that under standard complexity theoretical assumptions these problems do not admit a fixed-parameter tractable algorithm with respect to $k$, i.e., an algorithm running in $f(k)\lvert \mathcal{I} \rvert^{\mathcal{O}(1)}$ for some computable function~$f$.
Lastly, we provide efficient polynomial-time algorithms for \textsc{Plurality/Majority Delegation Bribery} in \Cref{sec:delegbribe}.

\section{Preliminaries}\label{sec:prelims}
For two sets $S$ and $S'$, we denote the symmetric difference between~$S$ and~$S'$ by~$S\triangle S'$.
\paragraph*{Graph theory.}
Let~$G = (V, A)$ be a directed graph.
We define~$\graphsize{G}\coloneqq \setcard{V} + \setcard{A}$ to be the size of~$G$.
Furthermore, we denote by~$\sinks(G)\subseteq V$ the set of sinks in~$G$, that is,~${\sinks(G)\coloneqq\{s\in V\mid \nexists v\in V: (s, v)\in A\}}$.
For a graph~$G' = (V', A')$, we write~$G'\gsubset G$ if~$V=V'$,~$A'\subseteq A$, and $S(G)=S(G')$.
A directed graph~$G=(V,A)$ is called a \emph{directed acyclic graph (DAG)} if $G$ does not contain any cycles (a \emph{cycle} is a sequence of vertices $v_1,\dots,v_p$ for some $p\geq 2$ such that $(v_i,v_{i+1})\in A$ for all $i\in [p-1]$ and $(v_p,v_1)\in A$). 
A \emph{path} in $G$ is a sequence of distinct vertices $v_1,\dots,v_p$ such that $(v_i,v_{i+1})\in A$ for all $i\in [p-1]$; the length of such a path is $p-1$. 
The \emph{depth} of a DAG is the maximum length of a path in $G$.

\paragraph*{Reachability sets.}
For a vertex~$v\in V$, let~$\revbfs_G(v)$ be the set of vertices that can \emph{reach}~$v$ in~$G$, that is, $$\revbfs_G(v)\coloneqq \{u\in V \mid \text{there is a $u$-$v$-path in $G$} \}.$$
Observe that the set~$\revbfs_G(v)$ is easy to compute in linear time by a breadth-first-search on~$G$ with all arcs reversed.
Similarly, for a sink~$s\in\sinks(G)$, we define $\resrevbfs_G(s)$ to be the set of vertices that can reach~$s$, but no other sink~$s'\neq s$, i.e., $$\resrevbfs_G(s)\coloneqq\{u\in\revbfs_G(s)\mid u\notin \revbfs_G(s') \text{ for all } s'\in\sinks(G)\setminus \{s\}\}.$$ 
This set can also be computed in linear time by a slightly modified two-phase reverse-breadth-first-search, where in the first phase we start from all vertices in~$\sinks(G)\setminus\{s\}$ and mark all vertices that can reach a sink~$s'\neq s$ and in the second phase we start from~$s$ and only visit unmarked vertices.

\paragraph*{Election Graphs.}
In our model of liquid democracy, we have a set $X$ of alternatives and a set~$Y$ of agents. 
Each agent $y\in Y$ either specifies an alternative $x_y\in X$ for which the agent wants to vote directly or a subset $Y_y\subseteq Y$ of delegation options, containing all agents which~$y$~deems qualified to receive their vote.  
One can model this input as a graph where we have one vertex $v_y$ for each agent~$y\in Y$ and one vertex $u_x$ for each alternative $x\in X$.
For each agent~$y\in Y$ voting directly for alternative $x_y$ we add an arc from $v_y$ to $u_{x_y}$. 
Moreover, for each agent~$y\in Y$ who wants to delegate its vote to some agent from~$Y_y$, we add an arc from~$v_y$ to~$v_{y'}$ for each~$y'\in Y_y$. 
Thus, in the constructed graph the sink vertices correspond to alternatives and the non-sink vertices to agents. 
To prevent the unavoidable selection of delegation cycles, we assume that in the constructed graph each $v_y$ for some $y\in Y$ has a path to some sink, as otherwise the vote of of $v_y$ could never be delegated to an alternative and thus a delegation cycle necessarily arises.\footnote{Concerning this assumption, first note that there exist multiple scenarios where the election graph is cycle-free by design, e.g., when agents have qualifications and are only allowed to delegate to an agent with a higher qualification. In fact, all our hardness result also cover this scenario. Second, in case there exist agents without a path to a sink, their vote cannot be assigned to any alternative. Thus, their vote is simply lost and we can delete the agent from the graph.}
Formally, we work on so-called election graphs: A directed graph~$G = (V, A)$ is an \emph{\election{} (graph)} if every vertex~$v\in V$ has a path to at least one sink~$s\in\sinks(G)$. We call the vertices in~$\sinks(G)$ \emph{alternatives} and the remaining ones \emph{agents}.\footnote{Election graphs are a slightly more general class of graphs as the graph class that results from the above described scenarios in the absence of delegation cycles, as election graphs allow agents to have outgoing arcs to both (possibly multiple) alternatives and arcs to other agents; however, whether we allow for this will be irrelevant in our subsequent computational analysis, so we use the more general model.}

\paragraph*{Delegation Graphs.}
Given an election graph $G = (V, A)$, a final \emph{delegation} needs to be chosen, i.e., for each agent either a single alternative the agent directly votes for or a single agent the agent delegates their vote to needs to be specified:
We call a graph~${G' = (V', A')}$ a \emph{delegation graph} if $G'$ does not contain any cycles and each vertex has outdegree at most one.
We say that a delegation graph~${G' = (V', A')}$ is a delegation subgraph of~$G$ if~$G'\gsubset G$.  
Observe that as we require that in an \election{} graph every vertex has a path to a sink, a \delegation{} subgraph exists in every \election{} graph.
In a delegation graph~$G'=(V',A')$, each agent $v\in V'\setminus S(G')$ has a certain voting power~$\vweight_{G'}(v)$, i.e., its own vote and the number of votes that are delegated to it.
This is defined as $$\vweight_{G'}(v)\coloneqq \setcard{\revbfs_{G'}(v)}.$$

The number of votes an alternatives $s\in S(G')$ receives in delegation~$G'$ is $$\vweight_{G'}(s)\coloneqq \setcard{\revbfs_{G'}(s)\setminus\sinks(G')}.$$
Note that since the vertices in~$\sinks(G')$ are artificial vertices representing alternatives, they do not contribute a vote themselves.

In a delegation graph $G'=(V',A')$, we say that an agent $v\in V'\setminus S(G')$ \emph{directly votes} for some alternative $s\in S(G')$ if~$(v,s)\in A'$; otherwise, we say that $v$ \emph{delegates} to $w\in V'\setminus S(G')$ if~$(v,w)\in A'$.
Moreover, we say that an agent $v\in V'\setminus S(G')$ \emph{votes} for alternative $s\in S(G')$ in~$G'$ if there is a path from $v$ to $s$ in $G'$, i.e., the vote of $v$ is (transitively) delegated to an agent that votes directly for $s$. 

\paragraph*{Voting Rules.}
Let~$G' = (V', A')$ be a delegation graph.
Recall that the voting power of a sink vertex/an alternative in $G'$ corresponds to the number of votes the alternative receives in~$G'$. 
Given these vote counts, we use two different voting rules to decide who wins the election: 
An alternative~$s\in\sinks(G')$ is a \emph{majority winner} in $G'$ if~$\vweight_{G'}(s) \geq\frac{\lvert V'\setminus S(G') \rvert}{2}$.\footnote{It would also be possible to consider a generalization of the majority voting rule parameterized by a threshold parameter $t$, where an alternative is a winner if it receives at least $t$ votes. This rule could for instance be used in elections consisting of multiple rounds where alternatives need a certain number of votes to advance to the next round. All our results also hold for this variant.}  
Further, an alternative~$s\in\sinks(G')$ is a \emph{plurality winner} in $G'$ if~$\vweight_{G'}(s) = \max_{s'\in\sinks(G')}\vweight_{G'}(s')$.
Observe that if we have only two alternatives, then the plurality and majority winners coincide. \smallskip

We introduce the studied computational problems in the relevant sections. 

\section{Winner Determination}\label{sec:windet}
In this section, we study \windetproblem{}, the problem concerned with determining whether an alternative can be a winner in one or all delegations in an election graph.

\medskip

\problemdef{\windetgenerictext{}}
{An \election{} graph~$G = (V, A)$ and a vertex~$s\in\sinks(G)$.}
{Is~$s$ a majority/plurality winner in at least one/in all delegation graphs~${G'\gsubset G}$?}

\medskip

In \cref{subsec:windetforall}, we consider the \allv{} variant and show that requiring that an alternative wins in all delegations is quite restrictive, leading to polynomial-time algorithms for both majority and plurality. In \cref{subsec:windetexists}, we investigate the \onev{} variants. While the corresponding problem is solvable in linear time under the majority rule, we show NP-hardness for the plurality rule, even on DAGs with only three alternatives and maximum outdegree two.

\subsection{\Allv{} Variant}\label{subsec:windetforall}
In this section, we show that the \textsc{All} variants of \textsc{Winner Determination} is polynomial-time solvable for both plurality and majority.
Our polynomial-time algorithms for \textsc{Majority/Plurality-All Winner Determination} both crucially rely on the observation that given some election graph $G$ and alternative $s$ one can easily find both the smallest and the largest possible set of agents that can vote for $s$ in some delegation subgraph.
This observation is formalized in the following lemma.

\begin{restatable}{lemma}{votesupperbounddeleg}\label{thm:votesupperbounddeleg}
	Let~$G = (V, A)$ be an \election{} graph and~$s\in\sinks(G)$ be an alternative. 
	In every delegation subgraph~$G'$ of $G$ we have
	$$\setcard{\resrevbfs_{G}(s) \setminus S(G)} \leq \vweight_{G'}(s)\leq\setcard{\revbfs_G(s) \setminus S(G)}.$$
	These bounds are tight, that is, there exist delegation subgraphs~$G'$ and~$G^*$ of~$G$ with~${\vweight_{G'}(s) = \setcard{\revbfs_{G}(s) \setminus S(G)}}$ and~${\vweight_{G^*}(s) = \setcard{\resrevbfs_{G}(s) \setminus S(G)}}$.
\end{restatable}

\begin{proof}
	For the first part of the lemma, observe that by definition no agent outside~$\revbfs_{G}(s)$ has a path to~$s$ in~$G$.
	Thus only the votes of agents in~$\revbfs_{G}(s)$ can possibly vote for~$s$ and since the vertices in~$S(G)$ have no voting power themselves, the maximum number of votes for~$s$ is~$\setcard{\revbfs_{G}(s) \setminus S(G)}$.
	On the other hand, all vertices in~$\resrevbfs_{G}(s)$ have no path to any vertex in~$S(G) \setminus \{s\}$.
	Thus, in any delegation~$G'$ of~$G$ we have~$\vweight_{G'}(s) \geq \setcard{\resrevbfs_{G}(s) \setminus S(G)}$.
	
	For the second part of the lemma, consider first the graph induced by all vertices except for the ones in~$\revbfs_{G}(s)$.
	This is still an election graph and therefore has some delegation subgraph $F$.
	By definition, all vertices in~$\revbfs_{G}(s)$ have some path to~$s$ in~$G$, so in the graph induced by the vertices from $\revbfs_{G}(s)$ there is a delegation graph $F'$ where all agents from $\revbfs_{G}(s)$ delegate to $s$.
	Combining $F$ and $F'$ leads a delegation subgraph $G'$ of $G$ with  $\vweight_{G'}(s) =\setcard{\revbfs_{G}(s) \setminus S(G)}$.
	Second, consider the graph induced by all vertices except for the ones in~$\resrevbfs_{G}(s)$.
	Again, this is still an election graph and therefore has some delegation subgraph $F$.
	By definition, all vertices in~$\resrevbfs_{G}(s)$ have some path to~$s$ in~$G$ and no path to any other vertex~$s' \in S(G)\setminus \{s\}$.
	Thus, in any delegation subgraph~$G^*$ of~$G$ containing $F$, we have~${\vweight_{G^*}(s) = \setcard{\resrevbfs_{G}(s) \setminus S(G)}}$, concluding the proof.
\end{proof}

The polynomial-time algorithms now follow directly.
In case of the majority voting rule, we check whether an alternative~$s$ always wins by checking whether~$\setcard{\resrevbfs_G(s) \setminus S(G)} \geq \frac{\setcard{V \setminus S(G)}}{2}$, that is, whether the smallest set of agents voting for~$s$ contains at least half of the agents.
In case of the plurality voting rule, we can simply check whether the smallest possible set of agents voting for~$s$ is still larger than the greatest possible set for any alternative~$s'\neq s$.

\begin{restatable}{proposition}{windetallpluralityeasy}\label{thm:windetallpluralityeasy}
	\windetallmaj{} can be solved in~$\bigOh(\graphsize{G})$ time. \windetallplurality{} can be solved in~$\bigOh(\graphsize{G}\cdot\setcard{\sinks(G)})$ time.
\end{restatable}

\begin{proof}
	Let~$G = (V, A)$ be the input \election{} graph and let~$s\in\sinks(G)$ be the target alternative.
	
	\textbf{Majority.}
	By \cref{thm:votesupperbounddeleg}, the lowest number of votes that~$s$ can get in any delegation subgraph of $G$ is exactly~$\setcard{\resrevbfs_G(s) \setminus S(G)}$.
	Thus, in order to solve~\windetallmaj{}, we first compute~$\resrevbfs_G(s)$ in linear time and then simply return ``YES'' if and only if~${\setcard{\resrevbfs_G(s) \setminus S(G)} \geq\frac{\setcard{V \setminus S(G)}}{2}}$.
	\medskip
	
	\textbf{Plurality.}
	First, compute the smallest possible set of voters for~$s$, that is,~$\setcard{\resrevbfs_{G}(s)\setminus S(G)}$.
	Second, check for each other alternative~$s'\in S(G)\setminus \{s\}$ whether it could get more votes than~$s$, that is, whether~${\setcard{\revbfs_{G}(s') \setminus S(G)} > \setcard{\resrevbfs_{G}(s) \setminus S(G)}}$.
	
	Concerning the running time, since the set~$\resrevbfs_G(s)$ as well as all the sets~$\revbfs_G(s')$ for each $s'\in S(G)\setminus \{s\}$ can be computed in~$\bigOh(\graphsize{G})$ time each, the algorithm runs in~${\bigOh(\graphsize{G}\cdot\setcard{\sinks(G)})}$ time overall.
	
	For the correctness, note that by \Cref{thm:votesupperbounddeleg},~$s$ gets at least~$\setcard{\resrevbfs_G(s) \setminus S(G)}$ votes.
	Similarly, each \alternative{}~$s' \in S(G) \setminus \{s\}$ gets at most~$\setcard{\revbfs_G(s') \setminus S(G)}$ votes. 
	Furthermore, since~$\resrevbfs_G(s)$ and~$\revbfs_G(s')$ are disjoint for each~${s' \in S(G) \setminus \{s\}}$, there also exists a delegation where both of those bounds are achieved at the same time.
	This implies that our algorithm correctly checks whether there is an alternative~$s'$ that can get more votes than~$s$, i.e., it solves~\windetallplurality{}.
\end{proof}

\subsection{\Onev{} Variant}\label{subsec:windetexists}
In this subsection, we handle the \onev{} variants of \windetproblem{}. 
Observe that the problem is linear-time solvable under the majority rule:
By~\cref{thm:votesupperbounddeleg}, for some election graph $G$, $\setcard{\revbfs_{G}(s) \setminus S(G)}$ is the maximum number of votes that~an alternative $s$ can get in any delegation subgraph of $G$ and that a delegation in which this bound is achieved always exists.
Thus, to check whether a given alternative is a majority winner in at least one delegation, it suffices to compute~$\revbfs_{G}(s)$ in linear time and check whether~$\setcard{\revbfs_{G}(s) \setminus S(G)} \geq \frac{\setcard{V \setminus S(G)}}{2}$.
\begin{observation}\label{thm:majoritywinnereasy}
	\windetonemaj{} is solvable in~$\bigOh(\graphsize{G})$ time.
\end{observation}

We now turn to \windetoneplurality{}. 
Recall that the plurality and majority rule coincide for two alternatives. 
We show that in contrast to the majority rule, under the plurality rule deciding whether an alternative is a winner in at least one delegation is NP-hard, even for just three \alternatives{} and on DAGs with maximum outdegree two.  
Intuitively speaking, the difference in complexity between the two aggregation rules stems from the fact that if an \alternative{} has the majority of votes, it does not matter how the remaining votes are distributed between the other \alternatives{}.
On the other hand, in the case of the plurality rule, it may happen that an \alternative{} can win without getting a majority of votes, but only if the remaining votes are roughly evenly split between its opponents.
Finding the optimal way to split the remaining votes is NP-hard.
We make this explicit by introducing the following problem.

\medskip

\problemdef{Equal Power Delegation}
{An \election{} graph~$G = (V, A)$ with two sinks $s$ and $t$.}
{Is there a \delegation{} subgraph $G'$ of $G$ with~${\vweight_{G'}(s)=\vweight_{G'}(t)}$?}

\medskip

Notably, \textsc{Equal Power Delegation} is a special case of the \textsc{MinMaxWeight} problem introduced by G\"{o}lz et al.~\cite{golz2018fluid}, which asks for a delegation in which the voting power of no agent exceeds some given threshold.
The problem was proven to be NP-hard (on arbitrary election graphs with an arbitrary number of sinks) by G\"{o}lz et al.~\cite{golz2018fluid}. 
The following proposition significantly strengthens their result.

\begin{proposition}\label{pr:EqualPower}
	\textsc{Equal Power Delegation} is NP-hard even on DAGs with maximum outdegree two and depth two.
\end{proposition}
\begin{proof}
	We reduce from \textsc{Vertex Cover}, where given an undirected graph~$H=(W,E)$ and an integer $\ell$, the question is whether there is a subset~$W'\subseteq W$ of $\ell$ vertices such that each edge has at least one endpoint in $W'$. 
	Let~$\enc{H=(W,E), \ell}$ be an instance of \textsc{Vertex Cover}, where we assume without loss of generality that $\ell\leq \frac{\lvert W \rvert}{2}-1$ (if this is not the case we simply add sufficiently many isolated vertices). 
	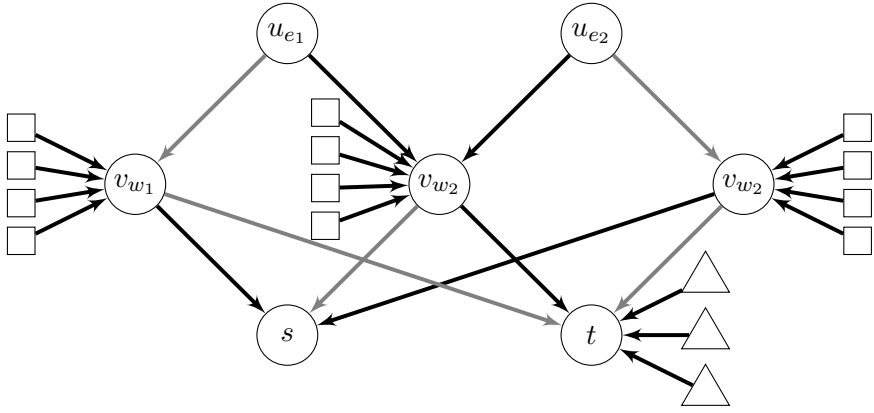
\begin{figure}[t]
		\centering
		\begin{tikzpicture}[main/.style = {draw, circle,minimum size=0.8cm,inner sep=0cm},square/.style={minimum size=0.4cm, draw,regular polygon,regular polygon sides=4},triangle/.style={minimum size=0.4cm, draw,regular polygon,regular polygon sides=3}] 
		\tikzset{edge/.style = {->,> = latex'}}
		\node[main] (a1) at (0,0) {$u_{e_1}$}; 
		\node[main] (a2) at (4,0) {$u_{e_2}$}; 
		
		\node[main] (b1) at (-2,-2) {$v_{w_1}$}; 
		\node[main] (b2) at (2,-2) {$v_{w_2}$};
		\node[main] (b3) at (6,-2) {$v_{w_2}$}; 
		
		\node[main] (d1) at (0,-4) {$s$}; 
		\node[main] (d2) at (4,-4) {$t$}; 
		
		\node[square] (c11) at (-3.5,-1.75) {};
		\node[square] (c12) at (-3.5,-2.25) {};
		\node[square] (c13) at (-3.5,-1.25) {};
		\node[square] (c14) at (-3.5,-2.75) {};
		
		\node[square] (c21) at (0.5,-1.55) {};
		\node[square] (c22) at (0.5,-2.05) {};
		\node[square] (c23) at (0.5,-1.05) {};
		\node[square] (c24) at (0.5,-2.55) {};
		
		\node[square] (c31) at (7.5,-1.75) {};
		\node[square] (c32) at (7.5,-2.25) {};
		\node[square] (c33) at (7.5,-1.25) {};
		\node[square] (c34) at (7.5,-2.75) {};
		
		\node[triangle] (f1) at (5.5,-4) {};
		\node[triangle] (f2) at (5.5,-4.75) {};
		\node[triangle] (f3) at (5.5,-3.25) {};
		
		\draw[edge,line width=1.5pt,gray] (a1) to (b1);
		\draw[edge,line width=1.5pt] (a1) to (b2);
		\draw[edge,line width=1.5pt] (a2) to (b2);
		\draw[edge,line width=1.5pt,gray] (a2) to (b3);
		\draw[edge,line width=1.5pt] (b1) to (d1);
		\draw[edge,line width=1.5pt,gray] (b2) to (d1);
		\draw[edge,line width=1.5pt] (b3) to (d1);
		\draw[edge,line width=1.5pt,gray] (b1) to (d2);
		\draw[edge,line width=1.5pt] (b2) to (d2);
		\draw[edge,line width=1.5pt,gray] (b3) to (d2);
		\draw[edge,line width=1.5pt] (c11) to (b1);
		\draw[edge,line width=1.5pt] (c12) to (b1);
		\draw[edge,line width=1.5pt] (c13) to (b1);
		\draw[edge,line width=1.5pt] (c14) to (b1);
		
		\draw[edge,line width=1.5pt] (c21) to (b2);
		\draw[edge,line width=1.5pt] (c22) to (b2);
		\draw[edge,line width=1.5pt] (c23) to (b2);
		\draw[edge,line width=1.5pt] (c24) to (b2);
		
		\draw[edge,line width=1.5pt] (c31) to (b3);
		\draw[edge,line width=1.5pt] (c32) to (b3);
		\draw[edge,line width=1.5pt] (c33) to (b3);
		\draw[edge,line width=1.5pt] (c34) to (b3);
		
		\draw[edge,line width=1.5pt] (f1) to (d2);
		\draw[edge,line width=1.5pt] (f2) to (d2);
		\draw[edge,line width=1.5pt] (f3) to (d2);
		\end{tikzpicture}
		\caption{Visualization of reduction from \Cref{pr:EqualPower} for the input graph $H=(W=\{w_1,w_2,w_3\},E=\{e_1=\{w_1,w_2\}, e_2=\{w_2,w_3\}\})$ and $\ell=1$. Dummy agents are drawn as squares and filling agents as triangles. The delegation corresponding to the vertex cover $\{w_2\}$ is drawn in black.}\label{fig:vis}
	\end{figure}
	From this, we create an instance $\enc{G=(V,A)}$ of \textsc{Equal Power Delegation} as follows (see \Cref{fig:vis} for a visualization of the construction). 
	We start by adding two sinks~$s$ and~$t$ to~$V$. 
	Moreover, for each $w\in W$, we add a \emph{vertex agent} $v_w$, and for each~$e\in E$, we add an \emph{edge agent} $u_e$. 
	Next, for each $\{w,w'\}=e\in E$, we add an arc from~$u_e$ to~$v_w$ and an arc from~$u_e$ to~$v_{w'}$. 
	In addition, for each~$w\in W$, we add an arc from~$v_w$ to~$s$ and an arc from~$v_w$ to~$t$. 
	Furthermore, for each $w\in W$, we introduce $\lvert W \rvert+\lvert E\rvert-1$ \emph{dummy agents} that only have an arc to $v_w$. 
	Lastly, we introduce $(\lvert W\rvert-2\ell)(\lvert W\rvert+\lvert E\rvert)-\lvert E\rvert\geq 0$ \emph{filling agents} that only have an arc to $t$. 
	Thus, overall, we have $2(\lvert W\rvert-\ell)(\lvert W\rvert+\lvert E\rvert)$ agents in $G$.
	Note that as all filling agents vote directly for $t$, we need to find a delegation in which exactly $\ell(\lvert W\rvert+\lvert E\rvert)+\lvert E\rvert$ of the vertex, dummy, and edge agents vote for $t$.
	We now show that~$\enc{H=(W,E), \ell}$ is a yes-instance of \textsc{Vertex Cover} if and only if  $\enc{G=(V,A)}$ is a yes-instance of \textsc{Equal Power Delegation}.
	
	$(\Rightarrow)$ Let $W'\subseteq W$ be a vertex cover of size $\ell$ in $H$.
	From this, we construct a delegation subgraph $G'$ of $G$ as follows. 
	For each $w\in W'$, vertex agent~$v_{w}$ directly votes for $t$, while for each $w\in W\setminus W'$, vertex agent~$v_{w}$ directly votes for $s$. 
	Lastly, for each~$\{w,w'\}=e\in E$, one of~$w$ and~$w'$, say~$w'$, is in~$W'$, since~$W'$ is a vertex cover. 
	We let $u_e$ delegate its vote to $v_{w'}$. 
	Thus, each edge agent delegates to a vertex agent directly voting for $t$ in~$G'$. 
	This means that in~$G'$ all~${(\lvert W \rvert-2\ell)(\lvert W \rvert+\lvert E \rvert)-\lvert E \rvert}$~filling agents, $\ell$ vertex agents together with~$\ell(\lvert W \rvert+\lvert E \rvert-1)$ dummy agents, and all~$\lvert E \rvert$~edge agents vote for $t$. 
	Thus, we have~${\vweight_{G'}(t)=(\lvert W \rvert-\ell)(\lvert W \rvert+\lvert E \rvert) =\vweight_{G'}(s)}$.
	
	$(\Leftarrow)$ Now assume that $G'$ is a delegation subgraph with ${\vweight_{G'}(t)=(\lvert W \rvert-\ell)(\lvert W \rvert+\lvert E \rvert) =\vweight_{G'}(s)}.$
	As observed above, in such a delegation exactly $\ell(\lvert W \rvert+\lvert E \rvert)+\lvert E \rvert$ of the vertex, dummy, and edge agents vote for $t$.
	Let~$W':=\{w\in W\mid v_w \text{ votes for $t$ in $G'$}\}$ be the vertices corresponding to vertex agents voting for $t$ in $G'$. 
	Recall that in all delegation subgraphs of $G$, for each vertex agent~$\lvert W \rvert+\lvert E \rvert-1$ dummy agents delegate to it and thus $\vweight_{G'}(v_w)\geq \lvert W \rvert+\lvert E \rvert$ for each $w\in W$.
	We claim that $W'$ is a vertex cover of size $\ell$ in $H$. 
	First note that it cannot hold that $\lvert W' \rvert>\ell$, as in this case at least $(\ell+1)(\lvert W \rvert+\lvert E \rvert)>\ell(\lvert W \rvert+\lvert E \rvert)+\lvert E \rvert$ vertex and dummy agents vote for $t$.
	Moreover,  it also cannot hold that $\lvert W' \rvert<\ell$, as in this case at most~$(\ell-1)(\lvert W \rvert+\lvert E \rvert)+\lvert E \rvert<\ell(\lvert W \rvert+\lvert E \rvert)+\lvert E \rvert$ vertex, dummy, and edge agents vote for $t$.
	Thus, we have $\lvert W' \rvert=\ell$, implying that $\ell(\lvert W \rvert+\lvert E \rvert)$ dummy and vertex agents vote for $t$ in $G'$. 
	Since, as observed above, $\ell(\lvert W \rvert+\lvert E \rvert)+\lvert E \rvert$ of the vertex, dummy, and edge agents vote for $t$ in $G'$, all edge agents vote for $t$ in $G'$. 
	As each edge agent $u_e$ for some $\{w,w'\}=e\in E$ either delegates to $v_w$ or $v_{w'}$ in $G'$ and only vertex agents corresponding to vertices from $W'$ vote for~$t$, this implies that each edge has at least one endpoint in $W'$. Thus, $W'$ is a vertex cover of $H$.
\end{proof}

Note that in the constructed election graph $G$ in the previous proof, we have that vertex agents have arcs to both sinks, which might be considered unrealistic from the perspective of liquid democracy. 
However, we can easily get rid of this by subdividing each arc to a sink by a newly introduced agent.

Using the hardness result above for deciding the existence of a delegation perfectly splitting the votes between two alternatives, it is straightforward to show the NP-hardness of~\windetoneplurality{}:
\begin{theorem}\label{thm:windetexpluralityhard}
	\windetoneplurality{} is NP-hard, even on DAGs with three \alternatives{} and maximum outdegree and depth two.
\end{theorem}
\begin{proof}
	We reduce from \textsc{Equal Power Delegation}.
	Let~$\enc{G = (V, A)}$ be an instance of \textsc{Equal Power Delegation} with sinks $s$ and~$t$ (without loss of generality we assume that $\lvert V \setminus S(G)  \rvert$ is even as otherwise we have a trivial no instance). 
	We now simply add a new sink $r$ and $\frac{\lvert V \setminus S(G)  \rvert}{2}$ dummy agents with an arc to~$r$. 
	Then, in a delegation subgraph $G'$ of $G$, $r$ will be a plurality winner if and only if~${\vweight_{G'}(s)=\vweight_{G'}(t)=\vweight_{G'}(r)=\frac{\lvert V \setminus S(G)  \rvert}{2}}$, which corresponds to~$G'$ restricted to~$G$ being a solution to the given \textsc{Equal Power Delegation} instance.
\end{proof}

\section{Election Bribery}\label{sec:elecbribe}
In this section, we handle~\electionbriberyproblem{}.
In this problem, we are given an \election{} graph, an alternative, and an integer $k$ and are asked to check whether it is possible for the given alternative to win after adding or deleting at most $k$ arcs (in all or at least one delegation).\footnote{Alternatively, one could also require that we are only allowed to change the outgoing arcs of at most $k$ agents; however, as we will argue at the end of this section, this does not change the computational complexity of our problems.}
The problem can be seen as a natural generalization of~\windetproblem{}: It asks whether some number~$k$ of modifications to a given election graph is enough to create a yes-instance of~\windetproblem{}.
In particular, for $k=0$, the two problems coincide.

\medskip

\problemdef{Majority/Plurality-One/All Election Bribery}
{An election graph~$G = (V, A)$, an alternative~$s\in\sinks(G)$, and an integer~$k \geq 0$.}
{Does there exist a set of arcs~$A^*\subseteq V\times V$ with~$\setcard{A\symmdiff A^*}\leq k$ such that~${G^* = (V, A^*)}$ is an \election{} graph with~$\sinks(G)=\sinks(G^*)$ and~$s$ is a majority/plurality winner in one/in all delegation subgraphs of $G^*$?}

\medskip

Observe that while the problem definition allows for both adding and removing arcs, it usually only makes sense to do one of the two. 
In the \onev{} variants, we want to ensure the \emph{existence} of a ``good'' delegation, so only adding arcs is useful.
The opposite is true for the \allv{} variants, with one exception: If we remove arcs from a vertex such that it is no longer connected to any sink, we have to insert an arc to connect the vertex to another alternative.
Observe that in both cases, if we insert arcs into the graph, then it is always optimal to only add arcs directly going to the given alternative.

From \Cref{thm:windetexpluralityhard} in  \Cref{subsec:windetexists}, it directly follows that \textsc{Plurality-One Election Bribery} is NP-hard for $k=0$, even on DAGs with three alternatives and with maximum outdegree and depth two.
In this section, we show that all four problems are in fact NP-hard and W[1]-hard with respect to $k$ for two alternatives even on DAGs with maximum outdegree and depth two.\footnote{Note that all variants of \textsc{Election Bribery} except \textsc{Plurality-One Election Bribery} are polynomial-time solvable for constant $k$, i.e., contained in XP: We can simply iterate over all possible sets of at most $k$ modifications, apply them, and subsequently solve the \textsc{Winner Determination} problem for the arising election graph in polynomial time.}
We start by showing hardness for the \textsc{All} variants. 

\begin{theorem}\label{thm:edgeallhard}
	\textsc{Majority/Plurality-All Election Bribery} are NP-hard and W[1]-hard parameterized by the budget $k$, even on DAGs with  two alternatives and maximum outdegree and depth two.
\end{theorem}
\begin{proof}
	The construction in this reduction is structurally related to the one from \Cref{pr:EqualPower}. 
	We reduce from \textsc{Clique}, where given a graph $H=(W,E)$ and an integer $\ell$ the question is whether there is a set of $\ell$ pairwise adjacent vertices.  
	\textsc{Clique} is NP-hard and W[1]-hard parameterized by $\ell$. 
	Let~$\enc{H = (W, E), \ell}$ be an instance of \textsc{Clique}, where we assume without loss of generality that~$\ell+{\ell \choose 2}\leq \frac{\lvert W \rvert+\lvert E \rvert}{2}$ (if this is not the case we simply add sufficiently many isolated vertices). 
	From this, we construct an instance $\enc{G=(V,A),s,k}$ of \textsc{Majority/Plurality-All Election Bribery} as follows. 
	We first add two sinks $s$ and $t$ to $V$. 
	Moreover, we add a \emph{vertex agent} $v_w$ for every $w\in W$ and an \emph{edge agent} $u_e$ for every $e\in E$. 
	For each edge~${e=\{w,w'\}\in E}$, we add an arc from~$u_e$ to $v_{w}$ and an arc from $u_e$ to $v_{w'}$. 
	We also add for each~$w\in W$ an arc from~$v_w$ to $s$ and an arc from $v_w$ to $t$.
	Lastly, we introduce $\lvert W \rvert+\lvert E \rvert-2(\ell+{\ell \choose 2})\geq 0$ \emph{filling agents}, each having an arc to $s$.
	We set our budget $k :=\ell$ and require that $s$ is a majority/plurality winner in all delegations. 
	Note that overall we have~$2\lvert W \rvert+2\lvert E \rvert-2(\ell+{\ell \choose 2})$ agents.
	
	$(\Rightarrow)$ Assume that $W'\subseteq W$ is a clique of size $\ell=k$ in $H$. 
	Then, we claim that in the election graph $G^*$ that results from $G$ after deleting for each $w\in W'$ the arc from $v_w$ to $t$, $s$ is a {majority/plurality} winner in all delegation subgraphs of~$G^*$. 
	We prove this by arguing that sufficiently many agents can only reach $s$ and not~$t$ in $G^*$ (and thus need to vote for~$s$ in all delegation subgraphs): 
	This holds for all $\lvert W \rvert+\lvert E \rvert-2(\ell+{\ell \choose 2})$ filling agents, for the $\ell$ vertex agents $\{v_w \mid w\in W'\}$, and for the ${\ell \choose 2}$ edge agents $\{u_{e}\mid e=\{w,w'\}\in E \wedge w,w'\in W'\}$. 
	Overall, these are $\lvert W \rvert+\lvert E \rvert-(\ell+{\ell \choose 2})$ agents which is exactly half of all agents. 
	Thus, $s$ is a majority/plurality winner in all delegation subgraphs of $G^*$. 
	
	$(\Leftarrow)$ Assume that 
	$A^*\subseteq V\times V$ is a set of arcs with~$\setcard{A\symmdiff A^*}\leq k$ such that~${G^* = (V, A^*)}$ is an \election{} graph with~$\sinks(G)=\sinks(G^*)$ and~$s$ is a majority/plurality winner in all delegation subgraphs of $G^*$.
	Without loss of generality, we can assume that no arcs have been added: 
	It clearly never makes sense to add any arcs to~$t$. 
	Moreover, the only agents without an arc to $s$ are edge agents. 
	Adding an arc from $u_{e}$ for some $e=\{w,w'\}\in E$ to $s$ can only be beneficial if we also delete the arc from $u_e$ to $v_w$ and $v_{w'}$. 
	However, in this case we can also simply delete the arc from~$v_w$ to $t$ and $v_{w'}$ to $t$ instead. 
	Thus, we can assume that $G$ and~$G^*$ only differ in arcs from some vertex agents to $t$. 
	Let~$W'=\{w\in W\mid (v_w,t)\notin A^*\}$ be the set of vertices from $H$ corresponding to the modified vertex agents. 
	We claim that $W'$ is a clique of size $\ell$ in $H$. 
	Note that in~$G^*$ there need to be~$\ell+{\ell \choose 2}$~more agents~$V'\subseteq V\setminus S(G)$ than in~$G$ which can only reach~$s$ and not~$t$. 
	Moreover, $V'$ clearly can only contain vertex agents $v_w$ with $w\in W'$ and edge agents~$u_e$ for $e=\{w,w'\}$ with $w,w'\in W'$. 
	As we have that $\lvert V' \rvert\geq \ell+{\ell \choose 2}$ and $\lvert W' \rvert\leq \ell$, this implies that there need to be ${\ell \choose 2}$ edges with both endpoints in $W'$, implying that $W'$ is a clique of size $\ell$ in~$H$.
\end{proof}

We now establish a similar hardness result for the \textsc{One} variants.
\begin{theorem}\label{thm:edgeonehard}
	\textsc{Majority/Plurality-One Election Bribery} are NP-hard and W[1]-hard parameterized by the budget $k$, even on DAGs with two alternatives and maximum outdegree and depth two.
\end{theorem}
\begin{proof}
	The construction again has some similarities to the construction from \Cref{pr:EqualPower}. 
	We reduce from \textsc{Independent Set}, where given an undirected graph~${H=(W,E)}$ and an integer $\ell$, the question is whether there is a subset~$W'\subseteq W$ of $\ell$ pairwise non-adjacent vertices. 
	\textsc{Independent Set} is NP-hard and W[1]-hard when parameterized by $\ell$ even on regular graphs. 
	Let $\enc{H = (W, E), \ell}$ be an instance of \textsc{Independent Set} where $H$ is $r$-regular for some~${r>0}$. 
	We assume that $\ell\leq \frac{\lvert W \rvert+\lvert E \rvert}{2(r+1)}$. 
	From this, we construct an instance $\enc{G=(V,A),s,k}$ of \textsc{Majority/Plurality-One Election Bribery} as follows. 
	We first add two sinks~$s$ and~$t$ to~$V$. 
	Subsequently, we introduce one \emph{vertex agent} $v_w$ for each $w\in W$ and one \emph{edge agent} $u_e$ for each $e\in E$. 
	We add an arc from $v_w$ to $t$ for each $w\in W$ and for each $\{w,w'\}=e\in E$ an arc from $u_e$ to $v_w$ and an arc from $u_e$ to $v_{w'}$.
	Lastly, we add~$\lvert W \rvert+\lvert E \rvert-2\ell(r+1)>0$ \emph{filling agents} which have only an arc to $s$.\footnote{If $\ell> \frac{\lvert W \rvert+\lvert E \rvert}{2(r+1)}$, then we need to slightly adjust the construction. Here instead of inserting~${\lvert W \rvert+\lvert E \rvert-2\ell(r+1)}$~\emph{filling agents} which have only an arc to $s$, we insert~${2\ell(r+1)-\lvert W \rvert-\lvert E \rvert\geq 0}$~\emph{filling agents} which have only an arc to $t$. The following proof of correctness works analogously in this case.}
	We set $k:=\ell$ and require that $s$ becomes a majority/plurality in at least one delegation subgraph of $G$. 
	In total, we have $2(\lvert W \rvert+\lvert E \rvert-\ell(r+1))$ agents in $G$. 
	
	$(\Rightarrow)$ Assume that $W'\subseteq W$ is an independent set of size $\ell$ in $H$. 
	We claim that in the election graph $G^*$ that results from $G$ after adding an arc from $v_w$ to $s$ for each~$w\in W'$, $s$ is a majority/plurality winner in at least one delegation subgraph. 
	In particular, $s$ is a winner in the delegation subgraph~$G'$ of~$G^*$ where each vertex agent $v_w$ for some $w\in W'$ votes directly for $s$ and each edge agent $u_e$ for some~$e=\{w,w'\}\in E$ with $w\in W'$ delegates to $v_w$. 
	We complete the delegation arbitrarily. 
	Then, in $G'$ all $\lvert W \rvert+\lvert E \rvert-2\ell(r+1)$ filling agents vote for $s$. 
	Moreover, all~$\ell$~vertex agents in~$\{v_w \mid w\in W'\}$ and all edge agents in~${\{u_e\mid \{w,w'\}=e\in E \wedge (w\in W' \vee w'\in W')\}}$ vote for $s$. 
	As $W'$ is an independent set of size $\ell$ in $H$ and $H$ is $r$-regular, we have
	$$\lvert \{u_e\mid \{w,w'\}=e\in E \wedge (w\in W' \vee w'\in W')\}\rvert=r\ell.$$
	Thus, in $G'$ at least $\lvert W \rvert+\lvert E \rvert-\ell(r+1)$ agents vote for $s$ and thus $s$ is a majority/plurality winner. 
	
	$(\Leftarrow)$ 
	Assume that 
	$A^*\subseteq V\times V$ is a set of arcs with~$\setcard{A\symmdiff A^*}\leq k$ such that~${G^* = (V, A^*)}$ is an \election{} graph with~$\sinks(G)=\sinks(G^*)$ and~$s$ is a majority/plurality winner in a delegation subgraph~$G'$ of~$G^*$.
	Without loss of generality, we can assume that $G^*$ and $G$ only differ in arcs from some agent to $s$. 
	Moreover, we can assume that no arc from some edge agent $u_e$ for some $\{w,w'\}=e\in E$ to $s$ has been added, as in this case we can simply add an arc from $v_w$ to $s$. 
	Let~$W'=\{w\in W \mid (v_w,s)\in A^*\}$ be the set of vertices from $H$ corresponding to the affected vertex agents. 
	Since we are only allowed to add~$k=\ell$~arcs, we clearly have $\lvert W' \rvert\leq \ell$.
	We claim that $W'$ is an independent set of size~$\ell$ in~$H$. 
	To see this, note that in $G'$ (and thus $G^*$) $\lvert W \rvert+\lvert E \rvert-\ell(r+1)$ agents need to have a path to~$s$. 
	As there are~$\lvert W \rvert+\lvert E \rvert-2\ell(r+1)$ filling agents, this implies that there need to be~$\ell(r+1)$~vertex and edge agents that have a path to $s$ in~$G^*$. 
	Note that an edge agent $u_e$ for~$\{w,w'\}=e\in E$ can only have a path to $s$ if~$w\in W'$ or $w'\in W'$ and that a vertex agent $v_w$ only has a path to $s$ if $w\in W'$. 
	Since $\lvert W' \rvert\leq \ell$ and each vertex in $H$ is incident to $r$ edges, it follows that no edge can have two endpoints in $W'$. 
	Moreover, we directly get that $\lvert W' \rvert=\ell$. 
	Thus, $W'$ is an independent set of size~$\ell$ in~$H$.
\end{proof}

\paragraph{Different cost models.}
In our definition of \textsc{Election Bribery} we ``pay'' per modified arc. 
Instead it would also be possible to pay per modified agent, i.e., pay for each agent whose outgoing arcs have been modified. 
Such different cost models are often studied in the context of bribery problems \cite{DBLP:journals/jair/BoehmerBHN21,DBLP:conf/ijcai/BoehmerBKL20,DBLP:journals/jair/FaliszewskiHHR09,faliszewski2016control}. 
Which of the two models is more suitable depends on the application. 
While paying per modified arc can model situations of indirect influence where the briber convinces an agent by discussions or advertisements, paying per modified agent models situations of direct influence where the briber ``buys'' agents and can thus completely control their behavior. 
However, for the computational complexity of our problems it does not matter which of the two cost models we use, as in both presented constructions it is never beneficial to alter more than one outgoing arc per agent. 
Thus, all our hardness results also hold for the agent cost model.

\section{Delegation Bribery}\label{sec:delegbribe}

In the previous section, we have considered manipulations of election graphs before the election.
In contrast, in this section, we focus on manipulating delegations after the election.
To this end, we start with a delegation graph and ask whether a given alternative can become a winner if we allow a given number~$k$ of agents to change their mind.

\medskip

\problemdef{\delegationbriberygeneric}
{A \delegation{} graph~$G' = (V, A')$, an alternative~$s\in S(G')$, and an integer~$k \geq 0$.}
{Does there exist a \delegation{} graph~$G^* = (V, A^*)$ such that~${\sinks(G') = \sinks(G^*)}$, $\setcard{A\symmdiff A^*}\leq 2k$, and~$s$~is a majority/plurality winner in $G^*$?}

\medskip

We bound the symmetric difference $\setcard{A\symmdiff A^*}$ by $2k$ here because changing the outgoing arc of an agent always contributes two to the symmetric difference.
Note that in the formulation of the problem we allow any agent to change the agent they delegate their vote to or to vote directly for some alternative. 
In fact, it is easy to see that it is always optimal to only alter originally directly voting agents and make them directly vote for $s$. 
Notably, our algorithms can easily be adapted to the setting where delegating agents can only switch to voting directly.

We start by considering \delegationbriberymajorityproblem.
Here, a simple greedy algorithm that bribes vertices in order of their voting power solves the problem in linear time.

\begin{proposition}\label{thm:delegbribemajeasy}
	\delegationbriberymajorityproblem{} can be solved in~$\bigOh(\setcard{V})$~time.
\end{proposition}
\begin{proof}
	Let~$\enc{G' = (V, A'), s, k}$ be an instance of~\delegationbriberymajorityproblem{}.
	We first compute the voting power~$\vweight_{G'}(\cdot)$ of all agents and then sort the directly voting agents in decreasing order of~$\vweight_{G'}(\cdot)$. 
	Finally, we bribe the $k$ directly voting agents with the highest voting power that do not already vote for $s$ to vote for $s$.
	If~$s$ is a majority winner afterwards, then we return ``YES''.
	Otherwise, we return ``NO''.
	
	For the correctness of the algorithm, first observe that since for any arc~$(u, v)\in A'$ we have~$\vweight_{G'}(u)\leq \vweight_{G'}(v)$, it is always optimal to bribe agents directly voting for an \alternative{}. 
	Furthermore, since in a \delegation{} each vertex has a unique path to a sink, bribing a directly voting agent has no influence on the voting power of a different directly voting agent.
	Thus, our algorithm computes the delegation maximizing the number of votes for $s$.
	
	For the running time, first note that computing~$\vweight_{G'}(v)$ for all agents~${v\in V\setminus S(G')}$ can easily be done in~$\bigOh(\setcard{V})$ time with a topological sort (recall that since~$G'$ is a delegation graph, we have~$\graphsize{G'} = \bigOh(\setcard{V})$). 
	Observe that sorting according to~$\vweight_{G'}(v)$ can also be done in linear time using bucket sort.
\end{proof}

For the plurality rule, a similar approach does not always work: It can be more important to weaken an opponent than to strengthen the designated candidate, as votes taken away from strong alternatives ``count double.''
For example, imagine that we have a delegation graph with three alternatives~$r$,~$s$, and~$t$ in which an agent with voting power $6$ votes directly for~$s$, an agent with voting power~$5$ votes directly for $t$ and three agents with voting power~$4$ each directly vote for~$r$. 
Then, if we bribe the agent with voting power~$5$ to vote for $s$, then $s$ will still loose against~$r$. 
However, if we bribe one of the agents with voting power~$4$ to vote for $s$, then $s$ becomes a plurality winner. 

Nevertheless, a slightly different greedy approach can be used to solve \delegationbriberypluralityproblem{}.

\begin{theorem}\label{th:delegation-plur}
	\delegationbriberypluralityproblem{} can be solved in~$\bigOh(\setcard{V}^2)$~time.
\end{theorem}
\begin{proof}
	Let~$\enc{G' = (V, A'), s, k}$ be an instance of~\delegationbriberypluralityproblem{}.
	We apply a ``guess-and-verify'' approach.
	We call a value~$p\in\NN$ \emph{achievable} if the alternative~$s$ can get at least~$p$ votes and all other alternatives~$s'\in S(G')\setminus \{s\}$ at most~$p$ votes after~$k$ bribes.
	Observe that~$\enc{G', s, k}$ is a yes-instance of~\delegationbriberypluralityproblem{} if and only if there exists an achievable~$p$.
	In order to solve the problem, we therefore iterate over all values of $p\in \{1,\dots \lvert V\setminus S(G')\rvert\}$ and then check whether it is achievable.
	
	In order to check whether a value~$p$ is achievable, we proceed in two steps.
	First, we ensure that all alternatives different from $s$ get at most~$p$ votes.
	For this, for each alternative~${s'\in S(G')\setminus \{s\}}$, we bribe agents directly voting for~$s'$ in decreasing order of~$\vweight_{G'}(\cdot)$ to vote for $s$ until the number of votes for alternative~$s'$ drops to at most~$p$.
	Second, we ensure that~$s$ gets at least~$p$ votes.
	To do so, we bribe agents still voting directly for some other alternative in decreasing order of their voting power to vote for $s$.
	After the two steps,~$s$ is clearly a plurality winner.
	Now let~$\ell$ be the total number of agents bribed in the two previous steps.
	If~$\ell\leq k$, then~$p$ is achievable.
	Otherwise,~$p$ is not achievable: 
	As bribing a directly voting agent has no influence on the voting power of other directly voting agents, $\ell$ is the smallest number of agents that need to be bribed in order to ensure that all alternatives different from~$s$ get at most~$p$ votes while $s$ gets at least $p$ votes.
	
	For the running time, observe that computing~$\vweight_{G'}(v)$ for all agents $v\in V\setminus S(G')$ can easily be done in~$\bigOh(\setcard{V})$ time with a topological sort. 
	Moreover, we can compute the ordering of directly voting agents according to~$\vweight_{G'}(\cdot)$ in linear time using bucket sort.
	Observe that the verification step takes~$\bigOh(\setcard{V})$ time.
	Since there are only~$\bigOh(\setcard{V})$ possible values for $p$, it follows that~\delegationbriberypluralityproblem{} can be solved in~$\bigOh(\setcard{V}^2)$ time.
\end{proof}

\section{Conclusion and Future Work}\label{sec:conclusions}
We have considered a liquid democracy setting where agents can specify multiple delegation options. 
We have analyzed different questions regarding winner determination and bribery in this context.
In particular, we have analyzed the problem of deciding whether an alternative wins in one/all delegation(s), whether we can change the preferences of few agents such that a given alternative wins in one/all delegation(s) and whether we can partially change a given delegation such that a different alternative wins. 
For ten computational problems canonically arising in this way, we have provided a complete picture of their computational complexity, obtaining both polynomial-time algorithms and NP-hardness results. While all our bribery problems are framed constructively, we want to remark that all of them also admit natural destructive variants. 
Notably, all our results straightforwardly extend to these destructive variants. 
However, this might not be the case for future problems to be studied.
There are multiple possible directions for future work emanating from our work. 

First, there also exist different models of liquid democracy where the type of questions studied in this paper are relevant. 
For instance, instead of simply specifying multiple acceptable delegation options, it is sometimes assumed that each agent provides a preference relation over all acceptable delegation options~\cite{brill2021liquid,DBLP:journals/mp/KavithaKMSS22}. 
Moreover, we could also assume that directly voting agents do not cast for a single candidate but provide us with a preference relation over alternatives.
These could then be aggregated using a more involved voting rule like Borda or Copeland voting. 

Second, while we have focused on decision problems, the counting variants of our problems would be of high practical interest. 
For instance, the counting variant of \textsc{Winner Determination} asks for the number of delegations where a given alternative is a winner. 
This information could be used to decide between two alternatives which both win in at least one delegation.
Similarly, it is also natural to study the counting variants of our bribery problems. 
This could be used to assess the robustness of winners to random changes \cite{DBLP:conf/eumas/BaumeisterH21,DBLP:conf/ijcai/BoehmerBFN21} and not to the ``worst-case'' changes we studied. 

Third, as different delegations are of different robustness according to our robustness measures induced by \textsc{Delegation Bribery}, it would be interesting to study the computational complexity of finding the least/most robust delegation. Furthermore, we could also search for the least/most robust delegation in which some given alternative wins to compare the robustness of these delegations as a further measure for the robustness of different winners.

Fourth, while we have focused on changing the election graph such that a given alternative wins in at least one/all delegation(s), it would also be natural to study the problem of changing the election graph such that the given alternative wins in a delegation which fulfills certain properties, e.g., which minimizes the maximum voting power of any directly voting agent or which is returned by some popular delegation-choosing algorithm. 

\subsection*{Acknowledgments}
NB was supported by the DFG project MaMu (NI 369/19) and by the DFG project ComSoc-MPMS (NI
369/22).

\bibliographystyle{splncs04}

\end{document}